\documentclass[letterpaper, 10 pt, conference]{ieeeconf}  

\IEEEoverridecommandlockouts                              

\title{An extremum seeking algorithm for monotone Nash equilibrium problems}
\author{Suad Krilašević and Sergio Grammatico%
\thanks{This work was partially supported by the ERC under research project COSMOS (802348). E-mail addresses: \{s.krilasevic-1, s.grammatico\}@tudelft.nl.}}
\date{October 2019}

\usepackage{graphicx}
\usepackage{amsmath}
\usepackage{amssymb}
\usepackage{mathtools}
\usepackage[utf8]{inputenc}

\newcommand{\nref}[1]{(\ref{#1})}
\newcommand{\bfs}[1]{\boldsymbol{#1}}
\newcommand{\col}[1]{\operatorname{col}\left({#1}\right)}
\newcommand{\diag}[1]{\operatorname{diag}\left({#1}\right)}
\newcommand{\kraj}{\hspace*{\fill} $\square$}

\newcommand{\m}[1]{\begin{bmatrix} #1 \end{bmatrix}}
\newcommand{\n}[1]{\left\| #1 \right\|}
\newcommand{\red}{\nonumber \\}
\newcommand{\vprod}[2]{\left \langle #1 \mid #2 \right \rangle}

\newcommand{\R}{\mathbb{R}}

\newcommand{\dom}{\mathrm{dom}}
\newcommand{\proj}{\mathrm{proj}}

\newcommand{\argmin}{\mathrm{argmin}}

\usepackage{xcolor}

\newtheorem{defn}{Definition}
\newtheorem{stan_assum}{Standing assumption}

\newtheorem{theorem}{Theorem}
\newtheorem{rem}{Remark}
\begin{document}

\maketitle

\begin{abstract}
    In this paper we consider the problem of finding a Nash equilibrium (NE) via zeroth-order feedback information in games with merely monotone pseudogradient mapping. Based on hybrid system theory, we propose a novel extremum seeking algorithm which converges to the set of Nash equilibria in a semi-global practical sense.  Finally, we present two simulation examples. The first shows that the standard extremum seeking algorithm fails, while ours succeeds in reaching NE. In the second, we simulate an allocation problem with fixed demand.
\end{abstract}

\section{Introduction}
Decision problems where self-interested intelligent systems or agents wish to optimize their individual cost objective function arise in many engineering applications, such as demand-side management in smart grids \cite{mohsenian2010autonomous}, \cite{saad2012game}, charging/discharging coordination for plug-in electric vehicles \cite{ma2011decentralized}, \cite{grammatico2017dynamic}, thermostatically controlled loads \cite{li2015market} and robotic formation control \cite{lin2014distributed}. The key feature that distinguishes these problems from multi-agent distributed optimization is the fact the cost functions are coupled together. Currently, one active research area is that of finding (seeking) actions that are self-enforceable, e.g. actions such that no agent has an incentive to unilaterally deviate from - the so-called Nash equilibrium (NE). Due to the coupling of the cost functions, on computing a NE algorithmically, some information on the actions of the other agents is necessary. The quality of this information can vary from knowing everything (full knowledge of the agent actions) \cite{yi2019operator}, estimates based on distributed consensus between the agents \cite{gadjov2019distributed}, to payoff-based estimates \cite{marden2009cooperative}, \cite{frihauf2011nash}. From the mentioned scenarios, the last one is of special interest as it requires no dedicated communication infrastructure.  \\ \\

\emph{Literature review:} In payoff-based algorithms, each agent can only measure the value of their cost function, but does not know its analytic form. Many of such algorithms are designed for NEPs with static agents with finite action spaces, e.g. \cite{goto2012payoff}, \cite{marden2009cooperative}, \cite{marden2012revisiting}. In the case of continuous action spaces, measurements are most often used to estimate the pseudogradient. A prominent class of payoff-based algorithms is Extremum Seeking Control (ESC). The main idea is to use perturbation signals to ``excite'' the cost function and estimate its gradient. Since the first general proof of convergence \cite{krstic2000stability}, there has been a strong research effort to extend the original ESC approach \cite{tan2006non}, \cite{ghaffari2012multivariable}, as well as to conceive diverse variants, e.g. \cite{durr2013lie}. ESC was used for NE seeking in \cite{frihauf2011nash} where the proposed algorithm is proven to converge to a neighborhood of a NE for games with strongly monotone pseudogradient \cite{krilavsevic2021learning}. The results are extended in \cite{liu2011stochastic} to include stochastic perturbation signals. In \cite{poveda2017framework}, Poveda and Teel propose a framework for the synthesis of a hybrid controller which could also be used for NEPs. The authors in \cite{poveda2020fixed} propose a fixed-time Nash equilibrium seeking algorithm which also requires a strongly monotone pseudogradient and communication between the agents. To the best of our knowledge, there is still no ESC algorithm for solving NEPs with merely \textit{monotone} pseudogradient. \\ \\ 
 A common approach for is to translate the NEP into a variational inequality (VI) \cite[Equ. 1.4.7]{facchinei2007finite}, and in turn into the problem of finding a zero of an operator \cite[Equ. 1.1.3]{facchinei2007finite}. For the special class of monotone operators, there exists a vast literature, see \cite{bauschke2011convex} for an overview. Each algorithm for finding zeros of monotone operators has different prerequisites and working assumptions that define what class of problems it can solve.  For example, the forward-backward algorithm requires that the forward operator, typically the pseudogradient, is monotone and cocoercive \cite[Thm. 26.14]{bauschke2011convex}, whilst  the forward-backward-forward algorithm requires only monotonicity of the forward operator \cite[Rem. 26.18]{bauschke2011convex}.  The drawback of the latter is that it requires two forwards steps per iteration, namely two (pseudo)gradient computations. Other algorithms such as the extragradient \cite{korpelevich1976extragradient} and the subgradient extragradient \cite{censor2011subgradient} ensure convergence with mere monotonicity of the pseudogradient, but still require two forward steps per iteration. Recently, the golden ratio algorithm in \cite{malitsky2019golden} is proven to converge in the monotone case with only one forward step. All of the previously mentioned algorithms are designed as discrete-time iterations. Most of them can be transformed into continuous-time algorithms, such as the forward-backward with projections onto convex sets  \cite{boct2017dynamical}, the forward-backward with projections onto tangent cones \cite{grammatico2017dynamic}, \cite{bianchi2021continuous}, forward-backward-forward \cite{bot2020forward} and the golden ratio \cite{gadjov2020exact}, albeit without projections  (see Appendix \ref{app: projection case}). \\ \\

\emph{Contribution}: Motivated  by  the  above  literature  and open  research  problem,  to  the  best  of  our  knowledge, we consider and solve for the first time the problem of learning (i.e., seeking via zeroth-order information) a NE in merely monotone games via ESC. Unlike other extremum seeking algorithms for NEPs, we do not require strong monotonicity of the pseudogradient. Specifically, we extend the results in \cite{gadjov2020exact} via hybrid systems theory to construct a novel extremum seeking scheme which exploits the single forward step property of the continuous-time golden ration algorithm. \\ \\

\emph{Notation}: $\mathbb{R}$ denotes the set of real numbers. For a matrix $A \in \mathbb{R}^{n \times m}$, $A^\top$ denotes its transpose. For vectors $x, y \in \mathbb{R}^{n}$, $M \in \R^{n \times n}$ a positive semi-definite matrix and $\mathcal{A} \subset \R^n$, $\vprod{x}{y}$, $\|x \|$, $\|x \|_M$ and $\|x \|_\mathcal{A}$ denote the Euclidean inner product, norm, weighted norm and distance to set respectively. Given $N$ vectors $x_1, \dots, x_N$, possibly of different dimensions, $\col{x_1, \dots x_N} \coloneqq \left[ x_1^\top, \dots, x_N^\top \right]^\top $. Collective vectors are defined as $\bfs{x} \coloneqq \col{x_1, \dots, x_N}$ and for each $i = 1, \dots, N$, $\bfs{x}_{-i} \coloneqq \col{ x_1, \dots,  x_{i -1},  x_{i + 1}, \dots, x_N } $. Given $N$ matrices $A_1$, $A_2$, \dots, $A_N$, $\operatorname{diag}\left(A_{1}, \ldots, A_{N}\right)$ denotes the block diagonal matrix with $A_i$ on its diagonal. For a function $v: \mathbb{R}^{n} \times \mathbb{R}^{m}  \rightarrow \mathbb{R}$ differentiable in the first argument, we denote the partial gradient vector as $\nabla_x v(x, y) \coloneqq \col{\frac{\partial v(x, y)}{\partial x_{1}}, \ldots, \frac{\partial v(x, y)}{\partial x_{N}}} \in \mathbb{R}^{n}$. We use $\mathbb{S}^{1}:=\left\{z \in \mathbb{R}^{2}: z_{1}^{2}+z_{2}^{2}=1\right\}$ to denote the unit circle in $\R^2$. The mapping $\proj_S : \R^n \rightarrow S$ denotes the projection onto a closed convex set $S$, i.e., $\proj_S(v) \coloneqq \argmin_{y \in S}\n{y - v}$.  
$\operatorname{Id}$ is the identity operator. $I_n$ is the identity matrix of dimension $n$ and $ \bfs{0}_n$ is vector column of $n$ zeros. A continuous function $\gamma: \R_+ \leftarrow \R_+$ is of class $\mathcal{K}$ if it is zero at zero and strictly increasing. A continuous function $\alpha: \R_+ \leftarrow \R_+$ is of class $\mathcal{L}$ if is non-increasing and converges to zero as its arguments grows unbounded. A continuous function $\beta: \R_+ \times \R_+ \rightarrow \R_+$ is of class $\mathcal{KL}$ if it is of class $\mathcal{K}$ in the first argument and of class $\mathcal{L}$ in the second argument.

\begin{defn}[UGAS]
For a dynamical system, with state $ x \in C \subseteq \R^n$ and 
\begin{align}
    & \dot{x} = f(x), \label{eq: reqular system}
\end{align}
where $f: \R^n \rightarrow \R^n$, a compact set $\mathcal{A} \subset C$ is uniformly globally asymptotically stable (UGAS) if there exists a $\mathcal{KL}$ function $\beta$ such that every solution of \nref{eq: reqular system} satisfies $\n{x(t)}_\mathcal{A} \leq \beta(\n{x(0)}_\mathcal{A}, t)$, for all $t \in \dom(x).$ \kraj
\end{defn}

\begin{defn}[SGPAS]
For a dynamical system parametrized by a vector of (small) positive parameters $\varepsilon \coloneqq \col{\varepsilon_1, \dots, \varepsilon_k}$, with state $x \in C \subseteq \R^n$ and 
\begin{align}
    \dot{x} = f_\varepsilon(x), \label{eq: small param system}
\end{align}
where $f_\varepsilon: \R^n \rightarrow \R^n$, a compact set $\mathcal{A} \subset C$ is semi-globally practically asymptotically stable (SGPAS) as $(\varepsilon_1, \dots, \varepsilon_k) \rightarrow 0^+$ if there exists a $\mathcal{KL}$ function $\beta$ such that the following holds:  For each $\Delta>0$ and $v>0$ there exists $\varepsilon_{0}^{*}>0$ such that for each $\varepsilon_{0} \in\left(0, \varepsilon_{0}^{*}\right)$ there exists $\varepsilon_{1}^{*}\left(\varepsilon_{0}\right)>0$ such that for each $\varepsilon_{1} \in$
$\left(0, \varepsilon_{1}^{*}\left(\varepsilon_{0}\right)\right) \ldots$ there exists $\varepsilon_{j}^{*}\left(\varepsilon_{j-1}\right)>0$ such that for each $\varepsilon_{j} \in$
$\left(0, \varepsilon_{j}^{*}\left(\varepsilon_{j-1}\right)\right) \ldots, j=\{2, \ldots, k\},$ each solution $x$ of \nref{eq: small param system} that satisfies $\n{x(0)}_\mathcal{A} \leq \Delta$ also satisfies $\n{x(t)}_\mathcal{A}  \leq \beta\left(\n{x(0)}_\mathcal{A} , t\right)+v$ for all $t \in \dom(x)$. \kraj
\end{defn}
\begin{rem}
In simple terms, for every initial conditions $x(0) \in \mathcal{A} + \Delta \mathbb{B}$, it is possible to tune the parameters $\varepsilon_1, \dots, \varepsilon_k$ in that order, such that the set $\mathcal{A} + v \mathbb{B}$ is UGAS.
\end{rem}
\section{Problem statement}
We consider a multi-agent system with $N$ agents indexed by $i \in \mathcal{I} = \{1, 2, \dots , N \}$, each with cost function 
\begin{align}
  J_i(u_i, \bfs{u}_{-i})\label{eq: cost_i},  
\end{align}
where $u_i \in  \mathbb{R}^{m_i}$ is the decision variable, $J_{i}: \mathbb{R}^{m_i} \rightarrow \mathbb{R}$. Let us also define $m \coloneqq \sum m_i$ and $m_{-i} \coloneqq \sum_{j \neq i} m_j$.\\ \\

In this paper, we assume that the goal of each agent is to minimize its cost function, i.e.,
\begin{align}
\forall i \in \mathcal{I}:\ \min_{u_i \in \R^{m_i}}  J_i(u_i, \boldsymbol{u}_{-i}), \label{def: dyn_game} 
\end{align}{}

which depends on the decision variables of other agents as well. From a game-theoretic perspective, this is the problem to compute a Nash equilibrium (NE), as formalized next.

\begin{defn}[Nash equilibrium]
A collection of decisions $\bfs{u}^*\coloneqq\col{u_i^*}_{i \in \mathcal{I}}$ is a Nash equilibrium if, for all $i \in \mathcal{I}$,
\begin{align}
    u_{i}^{*} \in \underset{v_{i} \in \mathbb{R}^{m_i}}{\operatorname{argmin}}\ J_{i}\left(v_{i}, \bfs{u}_{-i}^{*}\right).  \label{def: gne}
\end{align}\end{defn}{}
with $J_i$ as in \nref{eq: cost_i}.\kraj \\ \\

In plain words, a set of decision variables is a NE if no agent can improve its cost function by unilaterally changing its input. To ensure the existence of a NE and the solutions to the differential equations, we postulate the following basic assumption \cite[Cor. 4.2]{bacsar1998dynamic}:
\begin{stan_assum}[Regularity] \label{sassum: regularity}

For each $i \in \mathcal{I}$, the function $J_i$ in \nref{eq: cost_i} is continuously differentiable in $u_i$ and continuous in $\bfs{u}_{-i}$; the function $J_{i}\left(\cdot, \bfs{u}_{-i}\right)$ is strictly convex and radially unbounded in $u_i$ for every $\bfs{u}_{-i}$.\kraj
\end{stan_assum}

By stacking the partial gradients $\nabla_{u_i} J_i(u_i, \boldsymbol{u}_{-i})$ into a single vector, we form the so-called pseudogradient mapping
\begin{align}
    F(\boldsymbol{u}):=\operatorname{col}\left(\left(\nabla_{u_{i}} J_{i}\left(u_{i}, \bfs{u}_{-i}\right)\right)_{i \in \mathcal{I}}\right). \label{eq: pseudogradient}
\end{align}
From \cite[Equ. 1.4.7]{facchinei2007finite} and the fact that $u_i \in \R^{m_i}$ it follows that the Nash equilibrium $\bfs{u}^*$ satisfies 
\begin{align}
    F(\bfs{u}^*) = \bfs{0}_m. \label{eq: nash pseudogradient cond}
\end{align}

Let us also postulate the weakest working assumptions in NEPs with continuous actions, i.e. the monotonicity of the pseudogradient mapping and existence of solutions \cite[Def. 2.3.1, Thm. 2.3.4]{facchinei2007finite}.
\begin{stan_assum}[Monotonicity and existence]
The pseudogradient mapping $F$ in \nref{eq: pseudogradient} is monotone, i.e., for any pair $\bfs{u}, \bfs{v} \in \bfs{\Omega}$, it holds that $\vprod{\bfs{u} - \bfs{v}}{F(\bfs{u}) - F(\bfs{v})} \geq 0$; There exists a $\bfs{u}^*$ such that \nref{eq: nash pseudogradient cond} is satisfied.\kraj
\end{stan_assum}{}

Finally, let us define the following sets
\begin{align}
    \mathcal{S} &\coloneqq \{\bfs{u}\in \R^m \mid F(\bfs{u}) = 0\}\text{, (set of all NE)}\\
    \mathcal{A} &\coloneqq \{\col{\bfs{u}, \bfs{u}} \in \R^{2m} \mid \bfs{u}\in \mathcal{S}\}.
\end{align}
Thus, here we consider the problem of finding a NE of the game in \nref{def: gne} via the use of only zeroth-order information only, i.e. measurements of the values of the cost functions.

\section{Zeroth-order Nash  Equilibrium  seeking}
In this section, we present our main contribution: a novel extremum seeking algorithm for solving monotone NEPs. The extremum seeking dynamics consist of an oscillator $\bfs{\mu}$ which is used to excite the cost functions, a first-order filter $\bfs{\xi}$ that smooths out the pseudogradient estimation and improves performance, and a scheme as in \cite[Thm. 1]{gadjov2020exact} used for monotone NE seeking that, unlike regular pseudogradient decent, uses additional states $\bfs{z}$ in order to ensure convergence without strict monotonicity. We assume that the agents have access to the cost output only, hence, they do not directly know the actions of the other agents, nor they know the analytic expressions of their partial gradients. In fact, this is a standard setup used in extremum seeking, see (\cite{krstic2000stability}, \cite{poveda2017framework} among others. Our proposed continuous-time algorithm reads as follows

\begin{align}
     \forall i \in \mathcal{I}:  \m{
    \dot{z_i} \\
    \dot{u_i} \\
    \dot{\xi}\\
    \dot{\mu_i}
    } &=\m{
    \gamma_i\varepsilon_i\left(-z_i + u_i\right) \\
    \gamma_i\varepsilon_i\left(-u_i + z_i - \xi_i \right) \\
    \gamma_i\left( -\xi_i + \tilde F_i(\bfs{u}, \bfs{\mu}) \right) \\
    {2 \pi} \mathcal{R}_{i} \mu_i
    }, \label{eq: single agent dynamics}
\end{align}

or in equivalent collective form:

\begin{align}
    \m{
    \dot{\bfs{z}} \\
    \dot{\bfs{u}} \\
    \dot{\bfs{\xi}} \\
    \dot{\bfs{\mu}}
    }=\m{
    \bfs{\gamma}\bfs{\varepsilon}\left(-\bfs{z} + \bfs{\omega}\right) \\
    \bfs{\gamma}\bfs{\varepsilon}\left(-\bfs{\omega} + \bfs{z} - \bfs{\xi} \right) \\
    \bfs{\gamma}\left( - \bfs{\xi} + \tilde F(\bfs{u}, \bfs{\mu}) \right) \\
    {2 \pi} \mathcal{R}_{\kappa}\bfs{\mu}
    },\label{eq: complete system}
\end{align}

where $z_i, \xi_i \in \R^{m_{i}}$, $\mu_i \in \mathbb{S}^{m_i}$  are the oscillator states, $\varepsilon_i, \gamma_i \geq 0$ for all $i \in \mathcal{I}$, $\bfs{\varepsilon} \coloneqq \diag{ (\varepsilon_i I_{m_i})_{i \in \mathcal{I}}}$, $\bfs{\gamma} \coloneqq \diag{ (\gamma_i I_{m_i})_{i \in \mathcal{I}}}$, $\mathcal{R}_{\kappa} \coloneqq \diag{(\mathcal{R}_i)_{i \in \mathcal{I}}}$, $\mathcal{R}_i \coloneqq \diag{(\col{[0, -\kappa_j], [\kappa_j, 0]})_{j \in \mathcal{M}_i}}$, $\kappa_i > 0$ for all $i$ and $\kappa_i \neq \kappa_j$ for $i \neq j$, $\mathcal{M}_j \coloneqq \{\sum_{i = 1}^{j-1}m_i + 1, \dots, \sum_{i = 1}^{j-1}m_i + m_j\}$, $\mathbb{D}^n \in \R^{n \times 2n}$ is a matrix that selects every odd row from vector of size $2n$, $a_i > 0$ are small perturbation amplitude parameters, $A \coloneqq \diag{(a_i I_{m_i})_{j \in \mathcal{I}}}$, $J(\bfs{u}) = \diag{(J_i(u_i, \bfs{u}_{-i})I_{m_i})_{i \in \mathcal{I}}}$, $\tilde{F_i}(\bfs{u}, \bfs{\mu}) = \frac{2}{a_i}J_i(\bfs{u} + A \mathbb{D}^m \bfs{\mu}) \mathbb{D}^{m_i}{\mu_i}$ and $\tilde F(\bfs{u}, \bfs{\mu}) = 2A^{-1} J(\bfs{u} + A \mathbb{D}^m \bfs{\mu}) \mathbb{D}^m \bfs{\mu}$. Existence of solutions follows directly from \cite[Prop. 6.10]{goebel2009hybrid} as the the continuity of the right-hand side in \nref{eq: complete system} implies \cite[Assum. 6.5]{goebel2009hybrid}.\\ \\

Our main result is summarized in the following theorem.
\begin{theorem}
Let the Standing Assumptions hold and let $(\bfs{z}(t), \bfs{u}(t), \bfs{\xi}(t), \bfs{\mu}(t))_{t \geq 0}$ be the solution to \nref{eq: complete system} for arbitrary initial conditions. Then, the set $\mathcal{A} \times \R^m \times \mathbb{S}^m$ is SGPAS as $(\bar a, \bar\epsilon, \bar\gamma) = (\max_{i \in \mathcal{I}} a_i, \max_{i \in \mathcal{I}}\epsilon_i, \max_{i \in \mathcal{I}}\gamma_i) \rightarrow 0$. \kraj
\end{theorem}
\begin{proof}

We rewrite the system in \nref{eq: complete system} as

\begin{align}
    \m{
    \dot{\bfs{z}} \\
    \dot{\bfs{u}} \\
    \dot{\bfs{\xi}} \\
    }&=\m{
   \bar\gamma \bar\varepsilon \tilde{\bfs{\gamma}} \tilde{\bfs{\varepsilon}}\left(-\bfs{z} + \bfs{\omega}\right) \\
    \bar\gamma \bar\varepsilon \tilde{\bfs{\gamma}} \tilde{\bfs{\varepsilon}}\left(-\bfs{\omega} + \bfs{z} - \bfs{\xi}\right) \\
    \bar\gamma  \tilde{\bfs{\gamma}} \left( - \bfs{\xi} + \tilde F(\bfs{\omega}, \bfs{\mu}) \right)     },\label{eq: complete system rewritten} \\
    \dot{\bfs{\mu}} &=  {2 \pi} \mathcal{R}_{\kappa}\bfs{\mu}, \label{eq: oscilator}
\end{align}

 where $\bar\gamma \coloneqq \max_{i \in \mathcal{I}} \gamma_i$, $\bar\varepsilon \coloneqq \max_{i \in \mathcal{I}} \varepsilon_i$, $\tilde{\bfs{\gamma}} \coloneqq \bfs{\gamma}/{\bar \gamma}$ and $\tilde{\bfs{\varepsilon}} \coloneqq \bfs{\varepsilon}/{\bar \varepsilon}$.  The system in \nref{eq: complete system rewritten}, \nref{eq: oscilator} is in singular perturbation from where $\bar \gamma$ is the time scale separation constant. The goal is to average the dynamics of $\bfs{z}, \bfs{u}, \bfs{\xi}$ along the solutions of $\bfs{\mu}$. For sufficiently small $\bar a \coloneqq \max_{i \in \mathcal{I}} a_i$, we can use the Taylor expansion to write down the cost functions as 
\begin{align}
    &J_i(\bfs{u} + A \mathbb{D} \bfs{\mu}) = J_i(u_i, \bfs{u}_{-i}) + a_i (\mathbb{D}^{m_i} \mu_i)^\top \nabla_{u_i}J_i(u_i, \bfs{u}_{-i}) \red &\quad + A_{-i} (\mathbb{D}^{m_{-i}} \bfs{\mu}_{-i})^\top \nabla_{u_{-i}}J(u_i, \bfs{u}_{-i}) + O(\bar a^2), \label{eq: Taylor approx}
\end{align}
where $A_{-i} \coloneqq \diag{(a_i I_{m_i})_{j \in \mathcal{I} \setminus \{i\}}}$. By the fact that the right-hand side of \nref{eq: complete system rewritten}, \nref{eq: oscilator} is continuous, by using \cite[Lemma 1]{poveda2020fixed} and by substituting \nref{eq: Taylor approx} into \nref{eq: complete system rewritten}, we derive the well-defined average of the complete dynamics:

\begin{align}
    \m{
    \dot{\bfs{z}}^\text{a} \\
    \dot{\bfs{u}}^\text{a} \\
    \dot{\bfs{\xi}}^\text{a} 
    }=\m{
    \bar\varepsilon \tilde{\bfs{\gamma}} \tilde{\bfs{\varepsilon}}\left(-\bfs{z}^\text{a} + \bfs{u}^\text{a}\right) \\
     \bar\varepsilon \tilde{\bfs{\gamma}} \tilde{\bfs{\varepsilon}}\left(-\bfs{u}^\text{a} + \bfs{z}^\text{a} - \bfs{\xi}^\text{a} \right) \\
     \tilde{\bfs{\gamma}} \left( - \bfs{\xi}^\text{a} +  F(\bfs{u}^\text{a}) + \mathcal{O}(\bar a)\right) 
    }.\label{eq: average dynamics with O(a)}
\end{align}

The system in \nref{eq: average dynamics with O(a)} is an $\mathcal{O}(\bar a)$ perturbed version of the nominal average dynamics:

\begin{align}
    \m{
    \dot{\bfs{z}}^\text{a} \\
    \dot{\bfs{u}}^\text{a} \\
    \dot{\bfs{\xi}}^\text{a} 
    }=\m{
    \bar\varepsilon \tilde{\bfs{\gamma}} \tilde{\bfs{\varepsilon}}\left(-\bfs{z}^\text{a} + \bfs{u}^\text{a}\right) \\
     \bar\varepsilon \tilde{\bfs{\gamma}} \tilde{\bfs{\varepsilon}}\left(-\bfs{u}^\text{a} + \bfs{z}^\text{a} - \bfs{\xi}^\text{a}  \right) \\
     \tilde{\bfs{\gamma}} \left( - \bfs{\xi}^\text{a} +  F(\bfs{u}^\text{a})\right) 
    }.\label{eq: real nominal average dynamics}
\end{align}

For sufficiently small $\bar \varepsilon$, the system in \nref{eq: real nominal average dynamics} is in singular perturbation form with dynamics $\bfs{\xi}^\text{a}$ acting as fast dynamics. The boundary layer dynamics \cite[Eq. 11.14]{khalil2002nonlinear} are given by

\begin{align}
    \dot{\bfs{\xi}}^\text{a}_{\text{bl}} = \tilde{\bfs{\gamma}} \left( - \bfs{\xi}^\text{a}_{\text{bl}} +  F(\bfs{\omega}^\text{a}_{\text{bl}})\right) \label{boundary layer dynamics}
\end{align}

For each fixed $\bfs{\omega}^\text{a}_{\text{bl}}$, $\{F(\bfs{\omega}^\text{a}_{\text{bl}}))\}$ is an uniformly globally exponentially stable equilibrium point of the boundary layer dynamics. By \cite[Exm. 1]{wang2012analysis}, it holds that the system in \nref{eq: real nominal average dynamics} has a well-defined average system given by

\begin{align}
    \m{
    \dot{\bfs{z}}_\text{r} \\
    \dot{\bfs{u}}_\text{r}}=\m{
     \tilde{\bfs{\gamma}} \tilde{\bfs{\varepsilon}}\left(-\bfs{z}_\text{r} + \bfs{u}_\text{r}\right) \\
     \tilde{\bfs{\gamma}} \tilde{\bfs{\varepsilon}}\left(-\bfs{u}_\text{r} + \bfs{z}_\text{r} - F(\bfs{u}_\text{r})\right) 
    }. \label{eq: reduced system dynamics dynamics}
\end{align}
To prove that the system in \nref{eq: reduced system dynamics dynamics} renders the set $\mathcal{A}$ UGAS, we consider the following Lyapunov function candidate:
\begin{align}
    V(\bfs{u}_\text{r}, \bfs{z}_\text{r}) =  \tfrac{1}{2}\n{\bfs{z}_\text{r} - \bfs{u}^*}^2_{{\tilde{\bfs{\gamma}}^{-1} \tilde{\bfs{\varepsilon}}^{-1}}} + \tfrac{1}{2}\n{\bfs{u}_\text{r} - \bfs{u}^*}^2_{{\tilde{\bfs{\gamma}}^{-1} \tilde{\bfs{\varepsilon}}^{-1}}}. \label{eq: lyapunov fun cand}
\end{align}
The time derivative of the Lyapunov candidate is then
\begin{align}
    \dot{V}(\bfs{u}_\text{r}, \bfs{z}_\text{r}) &= \vprod{\bfs{z}_\text{r} - \bfs{u}^*}{-\bfs{z}_\text{r} + \bfs{u}_\text{r}} \red & \quad + \vprod{\bfs{u}_\text{r} - \bfs{u}^*}{-\bfs{u}_\text{r} + \bfs{z}_\text{r} - F(\bfs{u}_\text{r})} \red &= -\n{\bfs{u}_\text{r} - \bfs{z}_\text{r}}^2 - \vprod{\bfs{u}_\text{r} - \bfs{u}^*}{F(\bfs{u}_\text{r})} \red 
     &= -\n{\bfs{u}_\text{r} - \bfs{z}_\text{r}}^2 - \vprod{\bfs{u}_\text{r} - \bfs{u}^*}{F(\bfs{u}_\text{r}) - F(\bfs{u}^*)} \red
     & \leq \n{\bfs{u}_\text{r} - \bfs{z}_\text{r}}^2,\label{eq: imperfect vdot}
\end{align}
where the last two lines follow from \nref{eq: nash pseudogradient cond} and Standing Assumption 2. Let us define the following sets:
\begin{align}
    \Omega_c &\coloneqq \{(\bfs{u}_\text{r}, \bfs{z}_\text{r}) \in \R^{2m} \mid V(\bfs{u}_\text{r}, \bfs{z}_\text{r}) \leq c\} \red
    \Omega_0 &\coloneqq \{(\bfs{u}_\text{r}, \bfs{z}_\text{r}) \in \Omega_c \mid \bfs{u}_\text{r} = \bfs{z}_\text{r} \} \red
    \mathcal{Z} &\coloneqq \{(\bfs{u}_\text{r}, \bfs{z}_\text{r}) \in \Omega_c \mid \dot{V}(\bfs{u}_\text{r}, \bfs{z}_\text{r}) = 0\} \red
    \mathcal{O} &\coloneqq \{(\bfs{u}_\text{r}, \bfs{z}_\text{r}) \in \Omega_c \mid (\bfs{u}_\text{r}(0), \bfs{z}_\text{r}(0)) \in \mathcal{Z} \implies 
    \red
    &\quad \quad (\bfs{u}_\text{r}(t), \bfs{z}_\text{r}(t)) \in \mathcal{Z}\  \forall t \in \R\}, 
\end{align}
where $\Omega_c$ is a compact level set chosen such that it is nonempty, $\mathcal{Z}$ is the set of zeros of the Lyapunov function candidate derivative, $\Omega_0$ is the superset of $\mathcal{Z}$ which follows from \nref{eq: imperfect vdot} and $\mathcal{O}$ is the maximum invariant set as explained in \cite[Chp. 4.2]{khalil2002nonlinear}. Then it holds that
\begin{align}
    \Omega_c \supseteq \Omega_0 \supseteq \mathcal{Z} \supseteq \mathcal{O} \supseteq \mathcal{A}.
\end{align}
Firstly, for any compact set $\Omega_c$, since the right-hand side of \nref{eq: reduced system dynamics dynamics} is (locally) Lipschitz continuous and therefore by \cite[Thm. 3.3]{khalil2002nonlinear} we conclude that solutions to \nref{eq: reduced system dynamics dynamics} exist and are unique. Next, in order to prove convergence to a NE solution, we will show that $\mathcal{A} \equiv \mathcal{O}$, which is equivalent to saying that the only $\omega$-limit trajectories in $\mathcal{O}$ are the stationary points defined by $\mathcal{A}$. It is sufficient to prove that there cannot exist any positively invariant trajectories in $\Omega_0$ on any time interval $[t_1, t_2]$ where it holds
\begin{align}
    \m{\bfs{z}_\text{r}(t_1) \\ \bfs{u}_\text{r}(t_1)} \neq \m{\bfs{z}_\text{r}(t_2) \\ \bfs{u}_\text{r}(t_2)}\text{ and }  \m{\bfs{z}_\text{r}(t_1) \\ \bfs{u}_\text{r}(t_1)},\m{\bfs{z}_\text{r}(t_2) \\ \bfs{u}_\text{r}(t_2)} \in \Omega_0. \label{contradicion}
\end{align}
For the sake of contradiction, let us assume otherwise: there exists at least one such trajectory, which would be then defined by the following differential equations

\begin{align}
    \m{
    \dot{\bfs{z}_\text{r}} \\
    \dot{\bfs{u}_\text{r}}}=\m{
     0\\ - F(\bfs{u}_\text{r}) 
    }. \label{eq: invariant traj}
\end{align}

For this single trajectory, as $\dot{\bfs{z}_\text{r}} = 0$, it must then hold that $\bfs{z}_\text{r}(t_1) = \bfs{z}_\text{r}(t_2)$ and from the properties of $\Omega_0$, it follows that $\bfs{u}_\text{r}(t_1) = \bfs{z}_\text{r}(t_1)$ and $\bfs{u}_\text{r}(t_2) = \bfs{z}_\text{r}(t_2)$. From these three statements we conclude that $\bfs{u}_\text{r}(t_1) = \bfs{u}_\text{r}(t_2)$. Moreover, as a part of the definition of the trajectory, we must have $\bfs{u}_\text{r}(t_1) \neq \bfs{u}_\text{r}(t_2)$. Therefore, we have reached a contradiction. Thus, there does not exist any positively invariant trajectory in $\Omega_0$ such that \nref{contradicion} holds. Thus the only possible positively invariant trajectories are the ones where we have $\bfs{u}_\text{r}(t_1) = \bfs{u}_\text{r}(t_2)$ and $\bfs{z}_\text{r}(t_1) = \bfs{z}_\text{r}(t_2)$, which implies that $(\bfs{u}_\text{r}, \bfs{z}_\text{r}) \in \mathcal{A}$. Since the set $\mathcal{O}$ is a subset of the set $\Omega_0$, we conclude that the $\omega$-limit set is identical to the set $\mathcal{A}$. Therefore, by La Salle's theorem \cite[Thm. 4]{khalil2002nonlinear}, we conclude that the set $\mathcal{A}$ is UGAS for the in dynamics in \nref{eq: reduced system dynamics dynamics}.\\ \\

Next, by \cite[Thm. 2, Exm. 1]{wang2012analysis}, the dynamics in \nref{eq: real nominal average dynamics} render the set $\mathcal{A} \times \R^m$ SGPAS as $(\bar \varepsilon \rightarrow 0)$. As the right-hand side of the equations in \nref{eq: real nominal average dynamics} is continuous, the system is a well-posed hybrid dynamical system \cite[Thm. 6.30]{goebel2009hybrid} and therefore the $O(\bar a)$ perturbed system in \nref{eq: average dynamics with O(a)} renders the set $\mathcal{A} \times \R^m$ SGPAS as $(\bar \varepsilon, \bar a)\rightarrow 0$ \cite[Prop. A.1]{poveda2021robust}. By noticing that the set $\mathbb{S}^{m}$ is UGAS under oscillator dynamics in \nref{eq: oscilator} that generate a well-defined average system in \nref{eq: average dynamics with O(a)}, and by averaging results in \cite[Thm. 7]{poveda2020fixed}, we obtain that the dynamics in \nref{eq: complete system} make the set $\mathcal{A} \times \R^m \times \mathbb{S}^{m}$ SGPAS as $(\bar \varepsilon, \bar a, \bar \gamma) \rightarrow 0$. 

\end{proof}

\section{Simulation examples}
\subsection{Failure of the pseudogradient algorithm}
Let us consider a classic example on which the standard pseudogradient algorithm fails \cite{grammatico2018comments}:

\begin{align}
    J_1(u_1, u_2) &= (u_1 - u_1^*)(u_2 - u_2^*) \red
    J_2(u_1, u_2) &= -(u_1 - u_1^*)(u_2 - u_2^*), \label{eq: example}
\end{align}
where the game in \nref{eq: example} has a unique Nash Equilibrium at $(u_1^*, u_2^*)$ and the pseudogradient is only monotone. We compare our algorithm to the original and modified algorithm in \cite{frihauf2011nash}, where in the modified version, we introduce additional filtering dynamics to improve the performance:
\begin{align}
    \m{
    \dot{\bfs{u}} \\
    \dot{\bfs{\xi}} \\
    \dot{\bfs{\mu}}
    }=\m{
    -\bfs{\gamma}\bfs{\varepsilon} \bfs{\xi}  \\
    \bfs{\gamma}( - \bfs{\xi} + \tilde F(\bfs{u}, \bfs{\mu})) \\
    {2 \pi} \mathcal{R}_{\kappa}\bfs{\mu}
    }.
\end{align}

As simulation parameters we choose $a_i = 0.1$, $\gamma_i = 0.1$, $\varepsilon_i = 1$ for all $i$, $u_1^* = 2$, $u_2^* = -3$ and the frequency parameters $\kappa_i$ randomly in the range $[0, 1]$. We show the numerical results in Figures \ref{fig: u trajectories} and \ref{fig: u states}. The proposed NE seeking (NESC) algorithm steers the agents towards their NE, while both versions of the algorithm in \cite{frihauf2011nash} fail to converge and exhibits circular trajectories as in the full-information scenario \cite{grammatico2018comments}.

\begin{figure}
    \centering
    \includegraphics[width = \linewidth]{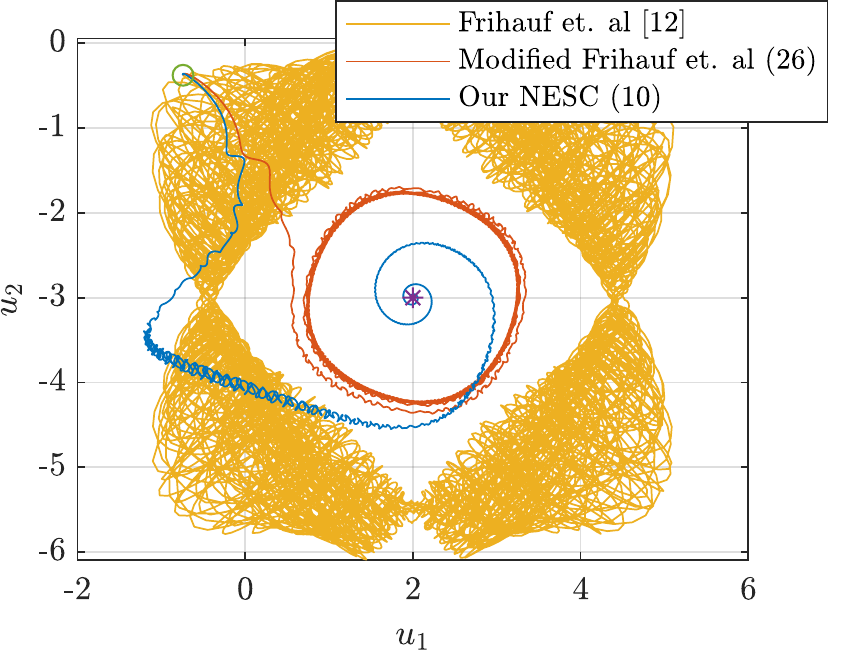}
    \caption{Input trajectories from an initial condition $\bfs{u}(0)$ (denoted by {\color{green}$\circ$}) towards $\bfs{u}^*$ (denoted by {\color{violet} $*$}).}
    \label{fig: u trajectories}
\end{figure}

\begin{figure}
    \centering
    \includegraphics[width = \linewidth]{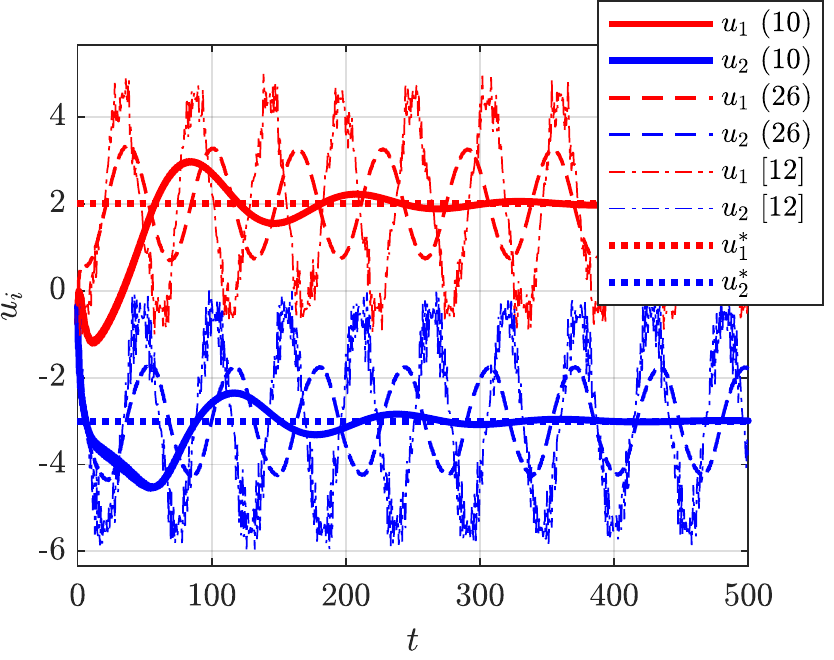}
    \caption{Time evolution of inputs $u_1$ and $u_2$ for the proposed NESC algorithm (solid line), the original algorithm in \cite{frihauf2011nash} (dash dot line) and the filtered version of the algorithm in \cite{frihauf2011nash} (dashed lines). }
    \label{fig: u states}
\end{figure}

\subsection{Fixed demand problem}
In the second simulation example, we consider the problem of determining the production outputs $u_i \in \R$ so that $N$ producers minimize their cost and meet the some fixed demand $U_d \in \R$ (see the power generator examples in \cite{aunedi2008optimizing} \cite{pantoja2011population}). The producers do not know the exact analytic form of their cost functions, which are given by:
\begin{align}
    J_i(u_i, \lambda) = u_i(u_i - 2U_i) - \lambda u_i,
\end{align}
where the first part corresponds to the unknown part of their cost and $\lambda u_i$ corresponds to the profit made by selling the commodity at the price $\lambda$. The last agent in this game is the market regulator whose goal is to balance supply and demand via the commodity price, by adopting the following cost function:
\begin{align}
    J_{N+1}(\bfs{u}, \lambda) =   \lambda (-U_d + \Sigma_{i = 1}^N u_i). 
\end{align}
The producers and the market regulator use the algorithm in \nref{eq: single agent dynamics} to determine the production output and price, albeit the market regulator uses the real value of its gradient, measurable as the discrepancy between the supply and demand. In the simulations, we use the following parameters: $N = 3$, $(U_1, U_2, U_3) = (172, 47, 66)kW$, $a_1 = a_2 = a_3 = 20$, $(\kappa_1, \kappa_2, \kappa_3) = (0.1778, 0.1238, 0.1824)$, $\epsilon_i = \tfrac{1}{3}$, $\gamma_i = 0.02$ for all $i$, $U_d = 350$ and zero initial conditions. In Figure \ref{fig:second_examplel}, we observe that the agents converge to the NE of the game. Additionally, we test the sensitivity of the commodity price with respect to additive measurement noise that obeys a Gaussian distribution with zero mean and standard deviation $\sigma$ for all producers. For different values of the standard deviation $\sigma$, we perform 200 numerical simulations. Next, we take the last 250 seconds of each simulation and sample it every 1 second. We group the resulting prices $\lambda(t)$ into bins of width 0.05, and plot the frequency of each bin. Three such plots are shown on Figure \ref{fig:second_example_distribution}. We observe that the price frequency plots seem to follow a Gaussian-like distribution. 


\begin{figure}
    \centering
    \includegraphics[width=\linewidth]{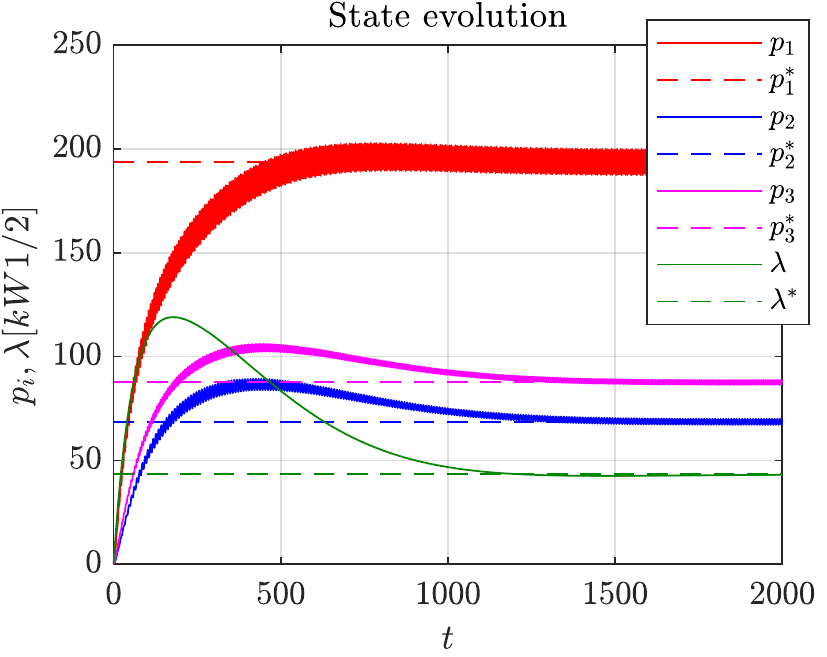}
    \caption{Time evolution of states in the fixed demand problem.}
    \label{fig:second_examplel}
\end{figure}
\begin{figure}
    \centering
    \includegraphics[width=\linewidth]{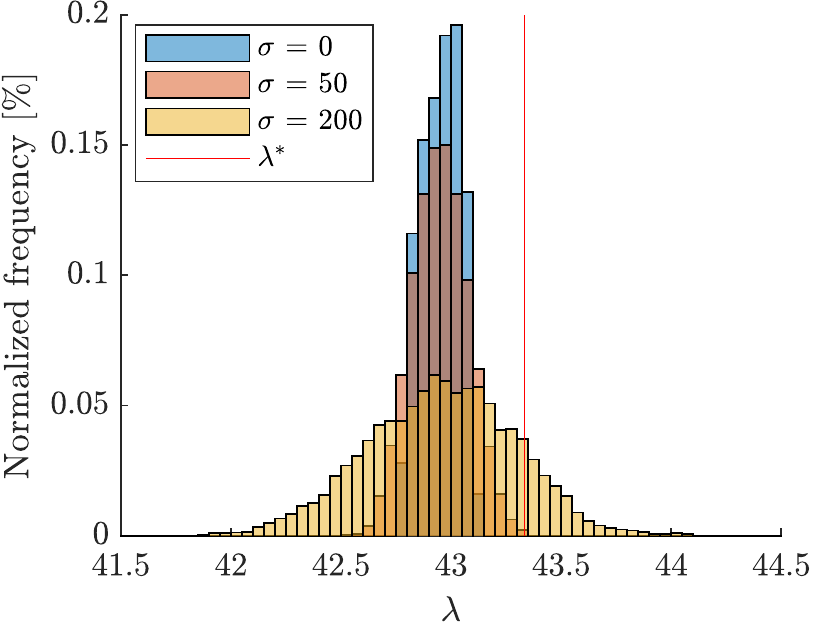}
    \caption{Price frequency distribution for three different cases.}
    \label{fig:second_example_distribution}
\end{figure}

\section{Conclusion}
Monotone Nash equilibrium problems without constraints can be solved via zeroth-order methods that leverage the properties of the continuous-time Golden ratio algorithm and ESC theory based on hybrid dynamical systems. 

\endlinechar=13
\bibliographystyle{IEEEtran}
\bibliography{biblioteka}

\endlinechar=-1
\appendix

\subsection{On the projection case} \label{app: projection case}
In this appendix, we show why the usual Lyapunov function candidate as in \nref{eq: lyapunov fun cand}, \cite{grammatico2017dynamic}, \cite{boct2017dynamical},  \cite{bianchi2021continuous}, \cite{bot2020forward}, cannot be used for the projected version of the proposed algorithm. In the simplest projected version, we would have projections onto convex sets as in \cite{malitsky2019golden}:
\begin{align}
    \m{
    \dot{\bfs{z}} \\
    \dot{\bfs{u}}}=\m{
     \left(-\bfs{z} + \bfs{u}\right) \\
     \left(-\bfs{u} + \proj_\Omega(\bfs{z} - F(\bfs{u})\right)) 
    }. \label{eq: nominal average dynamics bad}
\end{align}

Let us show a case where the Lyapunov function candidate in \nref{eq: lyapunov fun cand} increases. In Figure \ref{fig: counter_example}, we consider $F(\bfs{u}) = \col{u_2, -u_1}$, that the convex set $\bfs{\Omega}$ is given by the blue set and that the initial point is characterized by $\bfs{u} = \bfs{z} = [0, u_2]^\top$. The dotted lines represent the level sets of the Lyapunov function. The blue arrow represents the vector $-F(\bfs{u})$, while the red one represents $-\bfs{u} + \proj_\Omega(\bfs{z} - F(\bfs{u})$. We can see that the Lyapunov function does increase. For a different modulus of the pseudogradient vector, it is always possible to construct a convex set for which the Lyapunov function candidate increases. We conclude that a different Lyapunov candidate must be used to prove convergence in presence of constraints, which currently represents an open research problem.

\begin{figure}[ht]
    \centering
    \def\svgwidth{1\linewidth}
\begingroup%
  \makeatletter%
  \providecommand\color[2][]{%
    \errmessage{(Inkscape) Color is used for the text in Inkscape, but the package 'color.sty' is not loaded}%
    \renewcommand\color[2][]{}%
  }%
  \providecommand\transparent[1]{%
    \errmessage{(Inkscape) Transparency is used (non-zero) for the text in Inkscape, but the package 'transparent.sty' is not loaded}%
    \renewcommand\transparent[1]{}%
  }%
  \providecommand\rotatebox[2]{#2}%
  \newcommand*\fsize{\dimexpr\f@size pt\relax}%
  \newcommand*\lineheight[1]{\fontsize{\fsize}{#1\fsize}\selectfont}%
  \ifx\svgwidth\undefined%
    \setlength{\unitlength}{246.27155445bp}%
    \ifx\svgscale\undefined%
      \relax%
    \else%
      \setlength{\unitlength}{\unitlength * \real{\svgscale}}%
    \fi%
  \else%
    \setlength{\unitlength}{\svgwidth}%
  \fi%
  \global\let\svgwidth\undefined%
  \global\let\svgscale\undefined%
  \makeatother%
  \begin{picture}(1,0.79381966)%
    \lineheight{1}%
    \setlength\tabcolsep{0pt}%
    \put(0,0){\includegraphics[width=\unitlength,page=1]{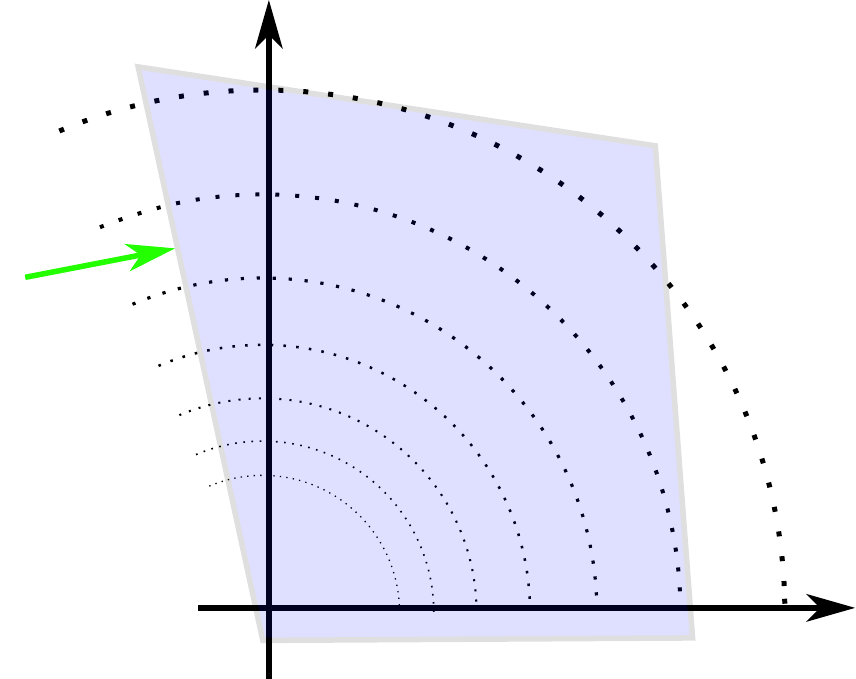}}%
    \put(0.34330384,0.60451692){\color[rgb]{0,0,0}\makebox(0,0)[lt]{\lineheight{1.25}\smash{\begin{tabular}[t]{l}$\Omega$\end{tabular}}}}%
    \put(0,0){\includegraphics[width=\unitlength,page=2]{counter_example.pdf}}%
    \put(0.33361341,0.1186243){\color[rgb]{0,0,0}\makebox(0,0)[lt]{\lineheight{1.25}\smash{\begin{tabular}[t]{l}$\bfs{u}^*$\end{tabular}}}}%
    \put(0,0){\includegraphics[width=\unitlength,page=3]{counter_example.pdf}}%
    \put(0.32557032,0.48671351){\color[rgb]{0,0,0}\makebox(0,0)[lt]{\lineheight{1.25}\smash{\begin{tabular}[t]{l}$(\bfs{u}, \bfs{z})$\end{tabular}}}}%
    \put(0,0){\includegraphics[width=\unitlength,page=4]{counter_example.pdf}}%
    \put(0.19190741,0.53066288){\color[rgb]{0,0,0}\makebox(0,0)[lt]{\lineheight{1.25}\smash{\begin{tabular}[t]{l}$\bfs{p}$\end{tabular}}}}%
  \end{picture}%
\endgroup%

    \caption{Blue arrow denotes the negative vector of $F(\bfs{u}) = \col{u_2, u_1}$, while the purple area represents the set $\Omega$; level sets of the Lyapunov function are represented with doted lines; $\bfs{u}^*$ is the NE; $\bfs{p} = \proj_\Omega(\bfs{u} - F(\bfs{u}))$. For the case when $\bfs{u} = \bfs{z}$, derivative $\dot{\bfs{u}} = -\bfs{u} + \bfs{p}$ is denoted with the red arrow. As the vector points outside of the level set, the Lyapunov function increases.}
    \label{fig: counter_example}
\end{figure}


\end{document}